\documentclass[11pt]{article}
\usepackage[latin9]{inputenc}
\usepackage{verbatim}
\usepackage{amsmath}
\usepackage{amsthm}
\usepackage{amssymb}
\usepackage{hyperref}

\usepackage{times}

\makeatletter
\theoremstyle{plain}
\newtheorem{thm}{\protect\theoremname}
\theoremstyle{definition}
\newtheorem{defn}[thm]{\protect\definitionname}
\theoremstyle{plain}
\newtheorem{lem}[thm]{\protect\lemmaname}
\ifx\proof\undefined
\newenvironment{proof}[1][\protect\proofname]{\par
	\normalfont\topsep6\p@\@plus6\p@\relax
	\trivlist
	\itemindent\parindent
	\item[\hskip\labelsep\scshape #1]\ignorespaces
}{%
	\endtrivlist\@endpefalse
}
\providecommand{\proofname}{Proof}
\fi
\theoremstyle{plain}
\newtheorem{cor}[thm]{\protect\corollaryname}

\usepackage{color}
\def\DEBUG{true}

\ifdefined\DEBUG
 \newcommand{\andy}[1]{\textcolor{blue}{#1}}
  \def\rem#1{{\marginpar{\raggedright\scriptsize #1}}}
  \newcommand{\andyr}[1]{\rem{\textcolor{blue}{$\bullet$ #1}}}
  
\else
  \newcommand{\andy}[1]{#1}
  \newcommand{\andyr}[1]{}
  
\fi

\newcommand*{\email}[1]{%
    \href{mailto:#1}{#1}\par
}

\usepackage{amsthm}

\newcommand{\jobs}{\mathcal J}
\newcommand{\machines}{\mathcal M}
\newcommand{\OPT}{\mathrm{OPT}}

\title{Minimizing Makespan with an Additive Error}

\usepackage{fullpage}

\makeatother

\providecommand{\corollaryname}{Corollary}
\providecommand{\definitionname}{Definition}
\providecommand{\lemmaname}{Lemma}
\providecommand{\theoremname}{Theorem}

\begin{document}
\global\long\def\N{\mathbb{N}}%
\global\long\def\K{\mathcal{K}}%

\title{Additive Approximation Schemes for Load Balancing Problems}
\author{%
Moritz Buchem\footnote{Maastricht University, Netherlands,
\email{m.buchem@maastrichtuniversity.nl}}
\and Lars Rohwedder\footnote{EPFL, Switzerland,
\email{lars.rohwedder@epfl.ch}, supported by the Swiss National Science Foundation project 200021-184656}
\and Tjark Vredeveld\footnote{Maastricht University, Netherlands,
\email{t.vredeveld@maastrichtuniversity.nl}}
\and Andreas Wiese\footnote{Universidad de Chile, Chile,
\email{awiese@dii.uchile.cl}}
}

\maketitle
\begin{abstract}
In this paper we introduce the concept of \emph{additive approximation schemes} and apply it to load balancing problems.
Additive approximation schemes aim to find a solution with an absolute error in the objective of at most $\epsilon h$ for some suitable parameter $h$.
In the case that the parameter $h$ provides a lower bound an additive approximation scheme implies a standard multiplicative approximation scheme and can be much stronger when $h \ll \OPT$.
On the other hand, when no PTAS exists (or is unlikely to exist), additive approximation schemes can provide a different notion for approximation.

We consider the problem of assigning jobs to identical machines with given lower and upper bounds for the loads of the machines. This setting generalizes problems like \emph{makespan minimization}, the \emph{Santa Claus} problem (on identical machines), and the \emph{envy-minimizing Santa Claus} problem.
For the last problem, in which the objective is to minimize the difference between the maximum and minimum load,
the optimal objective value may be zero and hence it is NP-hard to obtain any multiplicative approximation guarantee.
For this class of problems we present additive approximation schemes for $h = p_{\max}$, the maximum processing time of the jobs.

Our technical contribution is two-fold. First, we introduce a new relaxation based on integrally assigning slots
to machines and \emph{fractionally} assigning jobs to the slots. We refer to this relaxation as the \emph{slot-MILP}.
%
  We identify structural properties of (near-)optimal solutions of the slot-MILP, which allow us to solve it efficiently in polynomial time, assuming that there are $O(1)$ different lower and upper bounds on the machine loads (which is the relevant setting for the three problems mentioned above).
  The second technical contribution is a local-search based algorithm which rounds a solution to the slot-MILP
  introducing an additive error on the target load intervals of at most $\epsilon\cdot p_{\max}$.

\end{abstract}
\thispagestyle{empty}

\newpage{}

\setcounter{page}{1}

\section{Introduction}


In traditional analysis of approximation algorithms, one tries to
find a (multiplicative) guarantee $\rho$ such that the algorithm
finds a solution of value of at most (or at least, in case
of maximization problems) $\rho\cdot\OPT$, where $\OPT$ is the optimal
solution value. An approximation scheme is a family of 
approximation algorithms with performance guarantee of $\rho=(1+\epsilon)$
(or $(1-\epsilon)$ for maximization problems) for any $\epsilon>0$.
In this paper, we introduce the concept of \emph{additive approximation
schemes}. The goal is to design a family of algorithms that find a
solution with value not more than $\epsilon\cdot h$ away from the
optimal solution value, where $h$ is chosen to be a suitable parameter
of the problem instance. Formally, we define an additive approximation
scheme as follows. 
\begin{defn}
\label{def:additiveapprox} An \emph{additive approximation scheme}
is a family of algorithms that finds on any instance $I$ and for
every $\epsilon>0$ a solution with value $A(I)$ satisfying 
\[
\left|A(I)-\OPT(I)\right|\leq\epsilon h,
\]
where $h$ is a suitable chosen parameter of instance $I$. 
\end{defn}

In general, we are interested in finding additive approximation schemes that run in polynomial time, i.e., of the form $|I|^{f(1/\epsilon)}$, where $|I|$ denotes the size of the input and $f$ is some computable function.
Additive approximation schemes are particularly interesting in the
following two scenarios. 
\begin{enumerate}
\item When the problem at hand admits a PTAS and $h \ll \OPT$, one obtains a stronger guarantee than the PTAS.
\item When there cannot exist a PTAS, or even any multiplicative guarantee
for the problem, additive approximation schemes give an alternative
notion for approximating the problem. A notable example is the case when it is NP-complete to decide whether $\OPT = 0$, as then no multiplicative approximation guarantee can be obtained.
\end{enumerate}

Additive approximation has received only little attention in the literature.
Notable exceptions include Vizing's algorithm that finds an edge coloring
with at most $\Delta+1$ colors, where $\Delta$ is the maximum degree
of a graph~\cite{Vizing1964}. As $\Delta$ is a lower bound on the
minimum number of colors needed, this results implies an additive
$1$-approximation. Also, Alon et al.~\cite{AlonEtal2009} present
an additive $\epsilon n^{2}$ approximation algorithm for the edge
deletion problem to obtain a graph with a monotone property. This
falls into an additive approximation scheme for parameter $h=n^{2}$,
which is in fact an upper bound and not an lower bound on the minimum
number of edges to be deleted.

In this paper, we apply the concept of additive approximation schemes
to scheduling and load balancing problems which are
among the classical problems in the literature on approximation algorithms,
starting with the seminal work of Graham~\cite{Graham1966}. In these
problems $n$ jobs need to be processed by one of $m$ machines. A
job has processing time $p_{j}$ and the load of machine $i$ is the
sum of the processing times of the jobs assigned to $i$.
The goal is to find a schedule, which can be represented by an assignment
of jobs to the machines, that optimizes an objective function over
the machine loads. Since it is strongly NP-hard to decide whether
there is a schedule that assigns the same load to each machine (see~\cite{GareyJohnson1978}),
most non-trivial load balancing problems of this form are also strongly
NP-hard. This observation has led to extensive research on approximation
algorithms. In the following, we consider three variations of load
balancing problems.

\paragraph{Variations.}

The first objective function is to minimize the maximum machine load,
i.e., to minimize the \emph{makespan}. This is the one of the most
classical scheduling problems on parallel machines and has led to
the first approximation algorithms~\cite{Graham1966,Graham1969}.
Sahni~\cite{Sahni1976} showed that the problem admits an FPTAS for
constant number of machines and Hochbaum and Shmoys~\cite{HochbaumShmoys1987}
found a PTAS if the number of machines is part of the input.
Since then, there has been lively research in improving the running
time, e.g., to an EPTAS~\cite{Jansen2010,ChenEtal2013,JansenEtal2020}.

The second objective function that we consider is the \emph{Santa
Claus} problem, also known as \emph{max-min allocation}~\cite{Chakrabarty2008}.
Here, the goal is to maximize the minimum load, i.e., to make the
least loaded machine as full as possible. Bansal and Sviridenko~\cite{BansalSviridenko2006}
coined the term Santa Claus problem when they studied it in
the restricted assignment setting. This objective is considered
to measure the fairness of the allocation. The case of identical
machines was also consider by Woeginger~\cite{Woeginger1997} who
presents a PTAS.

As a third and final objective function, we consider to minimize the
\emph{maximum envy}, which is defined as the maximum load minus the
minimum load. This objective has been considered by Lipton et.~al~\cite{LiptonEtal2004}.
While in the Santa Claus problem fairness is measured by
the minimum load of a machine, in this setting fairness is considered
by the difference between the maximum and minimum load. 
Note that it is strongly NP-hard to decide whether or not the envy
is $0$. Therefore, unless $\text{P}=\text{NP}$, there cannot exist
any polynomial time approximation algorithm with any (multiplicative)
performance guarantee.

It is notable that for all three variants a simple greedy algorithm,
which assigns the jobs iteratively to the least loaded machine, gives
a additive error of $p_{\max}$. This guarantee is incomparable to
the error of $\epsilon\OPT$ of a PTAS.

\paragraph{Our contribution.}

In this paper, we will present additive approximation schemes for
load balancing problems on identical machines with parameter
$h=p_{\max}$. For the makespan and Santa Claus objective this gives
a significant improvement over the greedy algorithm mentioned above
while also dominating the guarantees of the known PTASs;
for minimizing the maximum envy this demonstrates how additive
approximation schemes can lead to non-trivial guarantees
when no multiplicative guarantees are possible.

For the mentioned load balancing problems this new perspective on
the analysis of approximation guarantees is particularly interesting,
because it requires fundamentally new methods: Most standard (multiplicative)
PTASes first round the processing times, find optimal solutions to
the rounded instance and then transform this into a solution for the
original instance. These methods do not work for additive approximation
scheme, because the rounding of all jobs directly adds an error of
$\epsilon\sum_{j}p_{j}$, which is too large compared to $\epsilon\cdot p_{\max}$.
Therefore, there is need for new (non-trivial) machinery. We present
a new relaxation for this general class of load balancing problems,
which we call the \emph{slot-MILP}. This slot-MILP can be
interpreted as a strenghened variant of the assignment-LP. The assignment-LP
is the relaxation that allows jobs to be assigned to the machines
fractionally. In the slot-MILP we first group the jobs of similar
size, but we do not round them (unlike previous PTASes). However,
in addition to the constraints of the assignment-LP we require an
integral number of jobs of each group to be assigned to each machine.
This property can be implemented using integer variables.
Alternatively, the relaxation can also be thought of assigning slots
(for the groups of jobs) \emph{integrally} to machines and then the jobs
\emph{fractionally} to the slots.
A straight-foward
application of Lenstra's algorithm fails, since the number of integer
variables is linear in the input size. Instead, we manage to solve
it using non-trivial structural properties combined with dynamic programming.
While the additive integrality gap of the assignment-LP can
be as large as $p_{\max}$, this gap is only $\epsilon\cdot p_{\max}$
for the slot-MILP. We show this using a rounding procedure inspired
by a local search method for the restricted assignment problem~\cite{DBLP:journals/siamcomp/Svensson12,DBLP:conf/soda/JansenR17,DBLP:conf/ipco/JansenR17}.
The local search algorithm repeatedly moves jobs between machines,
eventually converging to a good solution. Although in the restricted
assignment problem no polynomial running time bound is known for the
local search procedure, in our case we obtain such a bound
for our local search.

Our results extend to a more general setting in which for each
machine we are given a target interval for its load, with at most
$O(1)$ different such intervals across all machines. Our solution
then violates the desired load on each machine by at most $\pm\epsilon\cdot p_{\max}$
(or we assert that no solution exists for the given loads).

\paragraph*{Other related work.}
The case of small values of $p_{\max}$ has also been considered from
a parameterized point of view: If all processing times are integers,
then it is possible to obtain a running time that is fixed-parameter
tractable (FPT) in parameter $k = p_{\max}$~\cite{DBLP:journals/mp/MnichW15}.
In other words, there is an algorithm that finds
an exact solution in time $f(k) \cdot |I|^{O(1)}$
for some computable function $f$.

Other variants of load balancing problems have been considered in the paper by Alon et.~al~\cite{AlonEtal1998}. They identify some conditions on the objective function so that the results load balancing or machine schedule problem admits a PTAS.
Bansal and Sviridenko~\cite{BansalSviridenko2006} consider the Santa Claus problem in a more general setting where the processing time of a job $j$ is also dependent on the machine $i$ on which is processed, denoted by $p_{ij}$. They considered a restricted assignment setting in which each job is only allowed to be processed on a subset of the machines, but then the processing time is identical over all these machines, i.e., $p_{ij} \in \{p_j, 0\}$.
Bansal and Sviridenko~\cite{BansalSviridenko2006} presented an approximation algorithm with a performance guarantee dependent on the number of machines.
Feige~\cite{Feige2008} showed that the integrality gap of the configuration LP for this setting is constant using the Lovasz local lemma. Through the work of Moser and Tardos~\cite{MoserTardos2010} and Haeupler et al.~\cite{HaeuplerEtal2011} this could be turned into a
polynomial time algorithm. Asadpour et al.~\cite{Asadpour2012} present a local search method with a performance guarantee of $4$; however, it is unkown whether this method runs in polynomial time.

Related to additive approximation algorithms are several papers on the bin packing problem. Jansen et al.~\cite{JansenEtal2013} present an additive $1$-approximation algorithm in time exponential in the optimal number of bins plus a polynomial in the number of items to be packed. Hence, this algorithms is only useful when the optimal number of bins is small. On the other hand, Karmakar and Karp~\cite{KarmakarKarp1982} gave an algorithm that runs in polynomial time in the number of items giving a solution with at most $(1+\epsilon)\OPT + O(1/\epsilon^2)$ bins. Such a bound is called Asymptotic PTAS (APTAS) as the additive term vanishes when $\OPT$ is sufficiently large. This was subsequently improved by \cite{PlotkinEtal1995} and \cite{FernandezDeLaVegaLueker1981}.
Ophelders et al.~\cite{OpheldersEtal2020} showed that a simple local search algorithm for the so-called Equitable Hamiltonian Cycle finds a solution that is at most $1$ away from the optimal solution value.

\section{A new relaxation}
\label{sec:relax}
We introduce a new alternative relaxation for a general class of load balancing problems in machine scheduling. We first formally define this class of load balancing problems as the \emph{target load balancing problem}. 

\begin{defn}
In the \emph{target load balancing problem} we are given a set of
jobs $\jobs$ with a processing time $p_{j}$ for each $j\in\jobs$
and a set of machines $\machines$ with values $\ell_{i},u_{i}$ for
each machine $i\in\machines$. The goal is to assign each job $j\in\jobs$
to a machine $i\in\machines$ such that for each machine $i\in\machines$
the load of $i$ (i.e., the sum of the processing times of the jobs
assigned to $i$) is in the interval $[\ell_{i},u_{i}]$.
\end{defn}
This generalizes the load balancing settings mentioned earlier. For example, in $P||C_{max}$ every machine has a target load interval with $\ell_i = 0$ and $u_i = T$, where $T$ is a guess on the optimal makespan. 
For a given instance of the problem, we define $K$ to be the number
of different target load intervals $[\ell_{i},u_{i}]$ of
the machines in $\machines$, i.e., $K=\left|\left\{ [\ell_{i},u_{i}]|i\in\machines\right\} \right|$.
We will assume that $K=O(1)$.

Let $\epsilon>0$ and assume w.l.o.g.~that $1/\epsilon\in\N$. Our task is to
either assert that there is no solution
for the given instance or to find a solution in which the load of each
machine $i$ is in the interval $[\ell_{i}-\epsilon\cdot p_{\max},u_{i}+\epsilon\cdot p_{\max}]$
with $p_{\max}:=\max_{j\in J}p_{j}$, i.e.,
violating the target
load range of each machine by at most $\epsilon\cdot p_{\max}$. First
we partition the jobs into sets $\jobs_{1},\dotsc,\jobs_{1/\epsilon}$,
where for $k=1,...,1/\epsilon$ the set $\jobs_{k}$ contains
all jobs $j\in\jobs$ with $p_{j}\in((k-1)\epsilon\cdot p_{\max},k\epsilon\cdot p_{\max}]$.
We define a new relaxation for this problem in which for each machine
$i$ and each $k=1,...,1/\epsilon$ we specify integrally how
many jobs from $\jobs_{k}$ are assigned to $i$ (one may imagine
that this defines slots for jobs from $\jobs_{k}$ on $i$). Then
the jobs from $\jobs_{k}$ are assigned fractionally to these slots.
We denote by the slot-MILP the following relaxation.
\begin{align}
\min & \ 0\nonumber \\
\sum_{i\in\machines}x_{i,j} & =1 &  & \forall j\in\jobs\nonumber \\
\sum_{j\in\jobs}p_{j}x_{i,j} & \ge\ell_{i} &  & \forall i\in\machines\label{eq:lower-bound-load}\\
\sum_{j\in\jobs}p_{j}x_{i,j} & \le u_{i} &  & \forall i\in\machines\label{eq:upper-bound-load}\\
\sum_{j\in\jobs_{k}}x_{i,j} & =y_{i,k} &  & \forall i\in\machines,\forall k\in\{1,\dotsc,1/\epsilon\}\nonumber \\
x_{i,j} & \ge0 &  & \forall j\in\jobs,i\in\machines\nonumber \\
y_{i,k} & \in\N_{0}\text{ } &  & \forall i\in\machines,k\in\{1,\dotsc,1/\epsilon\}\nonumber 
\end{align}
In the slot-MILP the integer variables define exactly how many jobs of a type are assigned to a machine but do not imply a specific load based on rounded processing times. The load of a machine is based on an assignment that satisfies the distribution of slots among the machines.

Since the slot-MILP contains $1/\epsilon \cdot|\machines|$ integral variables, it
is not clear how to solve it in polynomial time. Nevertheless, we present two methods of efficiently solving the slot-MILP given that $K=O(1)$. The first method gives an exact solution while the second method gives a solution that slightly violates the target load intervals. Afterwards, we show how to round a fractional solution of the slot-MILP to an integral solution, while violating the load interval $[\ell_{i},u_{i}]$ for each machine $i\in\machines$ by at most $\epsilon\cdot p_{\max}$. In the following we give sketches of the proofs of the structural properties used to develop our solution methods. For detailed proofs we refer to Appendix~\ref{apx:relax_details}.

\subsection{Exact solution method for the relaxation}
\label{subsec:solving-relaxation}
We make use of a structural property to find an exact solution to the slot-MILP. Note that in this case an exact solution is one that satisfies~\eqref{eq:lower-bound-load} and~\eqref{eq:upper-bound-load}. This structure allows us to guess the values of the integral variables in polynomial time and then the remaining problem is only a linear program.

Given a solution $(x,y)$, for each machine $i\in\machines$ let $y_{i}$
denote the $(1/\epsilon)$-tuple $(y_{i,1},\ \dotsc,\ y_{i,1/\epsilon})$.
We show that there are solutions in which there are not too
many different vectors $y_{i}$.
This uses similar arguments to~\cite{DBLP:journals/orl/EisenbrandS06}.

\begin{lem}\label{lem:structural-property} There is a solution
$(x,y)$ to the slot-MILP such that for all $i,i'\in\machines$ with $[\ell_{i},u_{i}]=[\ell_{i'},u_{i'}]$
and $y_{i}\equiv y_{i'}\bmod2$ it follows that $y_{i}=y_{i'}$. 
\end{lem}
\begin{proof}[Proof sketch]
Let $(x,y)$ be a optimal solution to the slot-MILP and consider two machines $i_1,i_2$ with $[\ell_{i_1},u_{i_1}]=[\ell_{i_2},u_{i_2}]$. Suppose that $y_{i_1}\equiv y_{i_2}\bmod2$ but $y_{i_1}\neq y_{i_2}$. We construct a new solution $(x',y')$ which 
changes the jobs on machines $i_1$ and $i_2$ but
leaves the jobs of all other machines untouched. 
Intuitively, we assign to $i_1$ and $i_2$ the average load of both machines.
We define $x'_{i_{1},j}=x'_{i_{2},j}=(x_{i_{1},j}+x_{i_{2},j})/2$ for each $j\in \jobs$ and $y'_{i_1,k}=y'_{i_2,k}= (y_{i_1,k}+y_{i_2,k})/2$ for each $k\in \{1,...,1/\epsilon\}$.
Also, we set $x'_{i,j}=x_{i,j}$ and $y'_{i,k}=y_{i,k}$ for all $i \in \machines \setminus \{i_{1},i_{2}\}$, all $j\in \jobs$, and all $k\in \{1,...,1/\epsilon\}$.
As $y_{i}\equiv y_{i'}\bmod2$ and $(x,y)$ is feasible we have that $(x',y')$ is feasible as each job remains fully assigned and no machine is assigned more jobs of a type than it has slots. Furthermore, the load only changes on machines $i_1$ and $i_2$. However, as the new load on these machines becomes the average of the previous loads we have that all machine loads satisfy their respective target loads.
\end{proof}
For details of the proof we refer to Appendix~\ref{apx:pf_strucprop}.
Using Lemma~\ref{lem:structural-property} we can solve the slot-MILP in polynomial
time if $K=O(1)$.
\begin{lem}
\label{lem:solve-ILP}We can solve the slot-MILP in time $m^{O\left(K\cdot2^{1/\epsilon}\right)}\cdot n^{O\left(K/\epsilon\cdot2^{1/\epsilon}\right)}.$
\end{lem}
\begin{proof}
We first guess all values of $y_{i,k}$ (up to permutations of machines)
of the optimal solution due to Lemma~\ref{lem:structural-property} as
follows. We say that two machines $i,i'\in\machines$ are of the same
\emph{type }in a solution if $y_{i}=y_{i'}$ and $[\ell_{i},u_{i}]=[\ell_{i'},u_{i'}]$.
Then Lemma~\ref{lem:structural-property} implies that there are only
$K\cdot2^{1/\epsilon}$ different machine types. For each of these
$K\cdot2^{1/\epsilon}$ types we guess (1) the number of machines
having this type and (2) for each $k\in\{1,...,1/\epsilon\}$ we guess
the value of $y_{i,k}$ for each machine $i\in\machines$ of this
type. Note that the machines are identical and hence it suffices to
guess the number of machines of each type, rather than guessing which
exact machine is of which type. The total number of guesses is bounded
by $m^{O\left(K\cdot2^{1/\epsilon}\right)}\cdot n^{O\left(K/\epsilon\cdot2^{1/\epsilon}\right)}.$
Then the remaining problem is only a linear program (LP) since all
integral variables of the slot-MILP are already fixed. If our guess was correct
then the LP must have a feasible solution.
\end{proof}

\subsection{Faster (approximate) solution to the relaxation}
\label{sec:dp}
The solution based on Lemma~\ref{lem:structural-property} can be found in double exponential time with respect to the number of job types $1/\epsilon$ and is an exact solution to the slot-MILP. In the following we show that using a different (slightly more complicated) structural property one can find an additive $\delta$-approximate solution to the slot-MILP in single exponential time with respect to $1/\epsilon$ and polynomial in $1/\delta$, i.e., even with $\delta:=1/n^{O(1)}$ we obtain polynomial running time. Here, $\delta$-approximate implies that we find a solution to a weaker version of slot-MILP with $\ell'_i = \ell_i - \delta \cdot p_{\max}$ and $u'_i = u_i + \delta \cdot p_{\max}$ for every $i$. We refer to this weaker version as slot-MILP'.
The algorithm is based on a different structural property than the
one proved in Lemma~\ref{lem:structural-property}. Given a solution
$(x,y)$ of the slot-MILP, for each machine $i\in\machines$ and each $k\in\{1,\dotsc,1/\epsilon\}$,
we denote by $z_{i,k}$ the average size of the jobs type $k$ on machine $i$ defined by
\[
z_{i,k} \cdot y_{i,k}=\sum_{j\in\jobs_{k}}p_{j}x_{i,j}.
\]
In the case that $y_{i,k} = 0$ this allows us to freely choose the value of $z_{i,k}$ which is important for the structural property in the following lemma.
We prove that there is always a solution to the slot-MILP and an
ordering of the machines such that for each $k\in\{1,\dotsc,1/\epsilon\}$
the values $z_{i,k}$ are non-decreasing and on each prefix of length
$\ell$ of the machines the total size of the slots for the jobs in
$\jobs_{k}$ is at least as large as the $y_{\sigma(1),k}+\cdots+y_{\sigma(\ell),k}$
smallest jobs in $\jobs_{k}$.  For each integer $n'$ let $\jobs_{k}^{\min}(n')\subseteq\jobs_{k}$
be the $n'$ smallest jobs in $\jobs_{k}$.
\begin{lem}
\label{lem:structural-property-2} There is an optimal solution $(x,y)$
for the slot-MILP, a corresponding vector $\left\{ z_{i,k}\right\} _{i\in\machines,k\in1,\dotsc,1/\epsilon\}}$,
and an ordering $\sigma:\{1,...,|\machines|\}\rightarrow\machines$
such that 
\begin{align}
\sum_{\ell'=1}^{\ell}y_{\sigma(\ell'),k}z_{\sigma(\ell'),k} & \ge\sum_{j\in\jobs_{k}^{\min}(y_{\sigma(1),k}+\cdots+y_{\sigma(\ell),k})}p_{j} &  & \forall k\in\{1,\dotsc,1/\epsilon\}\,\forall\ell\in\{1,...,|\machines|\}\label{eq:min-jobs-bound}\\
\sum_{i\in\machines}y_{i,k}z_{i,k} & =\sum_{j\in\jobs_{k}}p_{j} &  & \forall k\in\{1,\dotsc,1/\epsilon\}\label{eq:volume-assigned}\\
z_{\sigma(\ell),k} & \le z_{\sigma(\ell+1),k} &  & \forall k\in\{1,\dotsc,1/\epsilon\}\,\forall\ell\in\{1,...,|\machines|-1\}.\label{eq:monotonicity}
\end{align}
\end{lem}
\begin{proof}[Proof sketch]
Conditions~\eqref{eq:min-jobs-bound} and~\eqref{eq:volume-assigned} follow from feasibility.
Condition~\eqref{eq:monotonicity} can be established by a potential function argument: we show that a solution minimizing this potential function has to fulfill condition~\eqref{eq:monotonicity} as otherwise we can swap some of the jobs of the same type between two machines and decrease the potential function while not decreasing the total load on the machines.
\end{proof}
We introduce a dynamic program that uses the property from Lemma~\ref{lem:structural-property-2}. Intuitively, our DP guesses the machines in the ordering $\sigma$
one after the other. When it guesses the next machine $i$, it first guesses the type of the machine, i.e. the values of $\ell_i$ and $u_i$, and then it guesses
for each $k\in\{1,\dotsc,1/\epsilon\}$ the value $z_{i,k}$ and the
number of jobs $y_{i,k}$ from $\jobs_{k}$ on machine $i$. In order to bound the running time we need to consider rounded values of $z_{i,k}$. Therefore, the DP ensures that the conditions~\eqref{eq:min-jobs-bound} and~\eqref{eq:monotonicity} on the vectors $y,z$ from
Lemma~\ref{lem:structural-property-2} are satisfied and that condition~\eqref{eq:volume-assigned} as well as the upper
and lower bounds on the load of each machine $i$ are only violated by a small extent. The following lemma shows that this is sufficient in order to compute
an approximate solution to the slot-MILP based on the vectors $y,z$.
\begin{lem}\label{lem:sufficient-conditions}
Suppose that we are given an ordering
$\sigma:\{1,...,|\machines|\}\rightarrow\machines$ and vectors $\left\{ y_{i,k},z_{i,k}\right\} _{i\in\machines,k\in\{1,\dotsc,1/\epsilon\}}$
such that conditions~\eqref{eq:min-jobs-bound} and~\eqref{eq:monotonicity} hold.  Moreover, assume that for each
$i\in\machines$ it holds that
\begin{equation}
\ell_{i}\le\sum_{k=1}^{1/\epsilon}y_{i,k}z_{i,k}\le u_{i} + \delta p_{\max} \label{eq:bounds}
\end{equation}
and for each $k \in \{1,\cdots,1/\epsilon\}$ we have that condition~\eqref{eq:volume-assigned} is slightly violated as follows
\begin{equation}
\sum_{j \in \jobs_k}p_j \leq \sum_{i \in \machines} y_{i,k}z_{i,k} \leq \sum_{j \in \jobs_k}p_j + \delta \epsilon\cdot p_{\max}. \label{eq:total_volume_weak}
\end{equation}
Then we can compute a vector $\left\{ x_{i,j}\right\} _{i\in\machines,j\in\jobs}$
such that $(x,y)$ is a solution to slot-MILP' in time $O(mn^2)$.
\end{lem}
\begin{proof}[Proof sketch]
We first find a (fractional) assignment vector $\left\{ x_{i,j}\right\} _{i\in\machines,j\in\jobs}$ satisfying
\begin{equation}
\sum_{j\in\jobs_{k}}p_{j}x_{i,j}\le y_{i,k}z_{i,k}\label{property-2-desired-main}
\end{equation}
for all $i \in \machines$ and $k \in \{1,\dots,1/\epsilon\}$. We do so by assigning jobs for each type $k$ independently. We start with a solution that assigns the smallest jobs to the first machine, and so on. Whenever a machine does not satisfy~\eqref{property-2-desired-main}, we can fractionally swap jobs between this machine and a machine with smaller index.
Once we have established that no machine is overloaded we establish that no machine is underloaded by too much, i.e. the following holds for all $i$ and $k$
\begin{equation}
    \sum_{j \in \jobs_k}p_jx_{i,j} \ge y_{i,k}z_{i,k} - \delta \epsilon\cdot p_{\max}. \label{eq:machine_under}
\end{equation}
\end{proof}
The goal of our DP is to compute vectors $\left\{ y_{i,k},z_{i,k}\right\} _{i\in\machines,k\in\{1,\dotsc,1/\epsilon\}}$
that satisfy the conditions due to Lemma~\ref{lem:sufficient-conditions}.
The key insight is now that when we consider the next machine $i'$
in the ordering, we do not need to remember all vectors $\left\{ y_{i,k},z_{i,k}\right\} _{i\in\machines,k\in\{1,\dotsc,1/\epsilon\}}$
for all previously considered machines~$i$, but it suffices to remember
the number of previously assigned jobs from each set $\jobs_{k}$,
the current left hand side of inequality~\eqref{eq:min-jobs-bound}
for each $k$, the vector $\left\{ z_{i'',k}\right\} _{k\in\{1,\dotsc,1/\epsilon\}}$
of the previously considered machine $i''$, and for each type of
machines the number of previously guessed machines of this type. At each iteration the machine then guesses the type of machine $i$ and the vectors $\left\{ y_{i,k},z_{i,k}\right\} _{i\in\machines,k\in\{1,\dotsc,1/\epsilon\}}$ such that the new solution consisting of the guess for machine $i$ and the remembered solution for the previous machines satisfies the conditions stated in Lemma~\ref{lem:solve-ILP-faster}. If none of the guesses satisfies these conditions the DP cell corresponding to this iteration remains empty. 

\begin{lem}
\label{lem:solve-ILP-faster}Let $\delta>0$. There is an algorithm
with a running time of $m^{K+1}(\frac{n}{\delta\epsilon})^{O(1/\epsilon)}$ which
either finds a $\delta$-approximate solution to the slot-MILP or asserts
that the slot-MILP is infeasible.
\end{lem}
The proof of Lemma~\ref{lem:solve-ILP-faster} and a detailed description of the dynamic program are given in Appendix~\ref{apx:dp_details}.

\section{\label{sec:rounding-relaxation}Rounding the relaxation}

We assume that we are given a solution to the slot-MILP via the algorithm
due to Lemma~\ref{lem:solve-ILP} or an approximate solution, i.e., a solution to slot-MILP' via the algorithm due to Lemma~\ref{lem:solve-ILP-faster}.
In this section, we describe an algorithm with a running time of $n^{O(1)}$
that computes an integral solution to the slot-MILP (or slot-MILP') which for each machine
$i\in\machines$ violates the target load by at most $\epsilon\cdot p_{\max}$. For a solution to the slot-MILP this implies that it holds that $\sum_{j\in\jobs}p_{j}x_{i,j}\in[\ell_{i}-\epsilon\cdot p_{\max},u_{i}+\epsilon\cdot p_{\max}]$. For a solution to slot-MILP' this implies that the target load violation is given by the error made due to the approximate solution and due to the rounding, i.e., after rounding the solution it holds that $\sum_{j\in\jobs}p_{j}x_{i,j}\in[\ell_{i}-\delta\cdot p_{\max}-\epsilon\cdot p_{\max},u_{i}+\delta \cdot p_{\max} + \epsilon\cdot p_{\max}]$. Note that the running time bound is independent of $K$. In the following we show how the rounding procedure works and that the claims hold for an exact solution to slot-MILP. The same arguments hold if the initial solution is a solution to slot-MILP'.

We imagine that each machine $i\in\machines$ has $y_{i,k}$ slots
for the jobs in $\jobs_{k}$, for each $k\in\{1,...,1/\epsilon\}$.
We say that these slots are of \emph{type $k$. }Notice that $\sum_{i\in\machines}y_{i,k}=|\jobs_{k}|$.
We compute an initial solution by assigning each job $j\in\jobs_{k}$
to an arbitrary slot of type $k$. In this solution there might be
a machine $i$ whose load is not in $[\ell_{i}-\epsilon\cdot p_{\max},u_{i}+\epsilon\cdot p_{\max}]$,
i.e., the load is too small or too large. We present now a local search
algorithm that repeatedly swaps pairs of jobs from the same set $\jobs_{k}$
such that eventually each machine $i\in\machines$ has a load in $[\ell_{i}-\epsilon\cdot p_{\max},u_{i}+\epsilon\cdot p_{\max}]$.
Note that this maintains the number of jobs from each set $\jobs_{k}$
on each machine.

\subsection{Local search}

We describe how to perform one iteration of the local search algorithm. Each iteration aims at finding a pair of jobs that can be swapped. Let $\machines_{1}$ be the set of machines $i\in\machines$ that have a load strictly greater than $u_{i}+\epsilon\cdot p_{\max}$.
Consider a $k\in\{1,...,1/\epsilon\}$ such that a job $j\in\jobs_{k}$
is assigned to a machine $i\in\machines_{1}$. We would like to exchange
$j$ for a smaller job $j'\in\jobs_{k}$ that is assigned to a machine
$i'\notin\machines_{1}$. Thus, consider all jobs $j'\in\jobs_{k}$
with $p_{j'}<p_{j}$ which are assigned to a machine $i'\notin\machines_{1}$.
If the load of $i'$ is at most $u_{i'}$ then we exchange $j$ and
$j'$ which completes the swap. We try to perform such a swap for
each $k\in\{1,...,1/\epsilon\}$ such that a job $j\in\jobs_{k}$
is assigned to a machine in $\machines_{1}$. If we did not perform
a swap then let $\machines_{2}$ denote the set of machines $i'\in\machines$
having a job $j'$ which we tried to swap with a job $j$ on a machine
$i\in\machines_{1}$, i.e., $\machines_{2}$ contains all machines
$i'\in\machines\setminus\machines_{1}$ for which there exists a machine
$i\in\machines_{1}$ and a $k\in\{1,...,1/\epsilon\}$ such that there
is a job $j\in\jobs_{k}$ assigned to $i$ and a job $j'\in\jobs_{k}$
assigned to $i'$ with $p_{j'}<p_{j}$. Observe that each machine
$i'\in\machines_{2}$ has a load of more than $u_{i'}$.

Now we repeat this procedure: Suppose that we constructed sets of machines
$\machines_{1},...,\machines_{\ell}$. For each $k\in\{1,...,1/\epsilon\}$
such that there is a job $j\in\jobs_{k}$ assigned to a machine $i\in\machines_{\ell}$
consider all jobs $j'\in\jobs_{k}$ with $p_{j'}<p_{j}$, which are
assigned to a machine $i'\notin\machines_{1}\cup\dotsc\cup\machines_{\ell}$.
If the load on one such machine $i'$ is at most $u_{i'}$, then we
exchange $j$ and $j'$ which completes the swap. In particular, we
do not reuse the constructed sets $\machines_{1},\dotsc,\machines_{\ell}$
for the next swap but we forget these sets before the next swap starts.
Otherwise, if each considered machine $i'$ has a load strictly more
than $u_{i'}$ we construct a set $\machines_{\ell+1}$ consisting
of all these machines $i'$ and continue in the current iteration.

Suppose that at the beginning of a swap there is no machine $i\in\machines$
that has a load strictly greater than $u_{i}+\epsilon\cdot p_{\max}$. Then, a second stage of the local search algorithm takes place. We take the current solution and perform an analogous procedure in order to ensure that each machine $i\in\machines$ has a load of
at least $\ell_{i}-\epsilon\cdot p_{\max}$. Initially define $\machines_{1}$
to be the set of all machines $i\in\machines$ with a load strictly
less than $\ell_{i}-\epsilon\cdot p_{\max}$. Suppose that we constructed
sets of machines $\machines_{1},...,\machines_{\ell}$. For each $k\in\{1,...,1/\epsilon\}$
such that there is a job $j\in\jobs_{k}$ assigned to a machine $i\in\machines_{\ell}$
consider all jobs $j'\in\jobs_{k}$ with $p_{j'}>p_{j}$, which are
assigned to a machine $i'\notin\machines_{1}\cup\dotsc\cup\machines_{\ell}$.
If the load on one such machine $i'$ is at least $\ell_{i'}$, then
we exchange $j$ and $j'$ which completes the swap. Otherwise, if
each such machine $i'$ has a load of strictly less than $\ell_{i'}$
we construct a set $\machines_{\ell+1}$ consisting of all these machines
$i'$ and continue. The algorithm terminates if the load of each machine
$i\in\machines$ is within $[\ell_{i}-\epsilon\cdot p_{\max},u_{i}+\epsilon\cdot p_{\max}]$. In the following we use first stage to refer to the part of the algorithm that establishes that all loads are at most $u_i + \epsilon \cdot p_{\max}$ and second stage to the part of the algorithm that establishes that all loads are at least $\ell_i - \epsilon \cdot p_{\max}$.

\subsection{Correctness and running time}
We show now that the algorithm terminates in $n^{O(1)}$ time. Then, by construction it outputs a solution in which each
machine $i\in\machines$ has a load in the interval $[\ell_{i}-\epsilon\cdot p_{\max},u_{i}+\epsilon\cdot p_{\max}]$.

We first show that in each iteration of the first stage of the algorithm we can find a pair of jobs $j,j'$ to swap. 

\begin{lem}\label{lem:ls_swaps}
In each iteration of the algorithm finds two jobs $j,j'$ that it swaps and finding such a pair can be done in $O(n^2)$.
\end{lem}
\begin{proof}[Proof]
We prove this for the first stage of the algorithm. A similar argument can be shown for the second stage (see Appendix~\ref{apx:det_ls}).
Suppose towards contradiction that the algorithm does not find two jobs $j,j' \in \jobs_k$ such that $j$ is assigned to a machine $i \in \machines_\ell$ with load more than $u_i+\epsilon\cdot p_{\max}$ and job $j'$ is assigned to a machine $i' \notin \machines_{1}\cup\dotsc\cup\machines_{\ell}$ and $p_j' < p_j$. This means that the machines in $\machines_{1}\cup\dotsc\cup\machines_{\ell}$ are assigned the smallest jobs of type $k$ while each machine having load at least $u_i$. Hence, even a fractional assignment of jobs cannot reduce the total load on the machines in $\machines_1 \cup \dots \cup \machines_\ell$. This means that in a fractional assignment at least one machine must have a load greater than $u_i$. This gives a contradiction. 

As every job is fully assigned to a machine we have that each job of a type $k$ which is assigned to a machine in the sets $\machines_1,\machines_2,\dots,\machines_\ell$ is considered exactly once and compared to every other job of type $k$ not assigned to these machines. As the existence of a pair was shown for an arbitrary $k$ this gives a worst case running time of $O(n^2)$.
\end{proof}
Next, we outline the proof that the first stage of the algorithm always terminates and that this happens after at most $O(n^3)$ swaps. As Lemma~\ref{lem:ls_swaps} states that each swap can be done in $O(n^2)$ this shows that the first stage finishes in $n^{O(1)}$. To this end, we give an alternative formulation of the first stage of the algorithm algorithm as a repeated breadth-first search (BFS).  We construct a weighted, directed graph. It contains one special vertex, the source $s$, and one vertex for each slot, that is $|\jobs|+1$
vertices in total. Each non-source vertex is associated with a machine
and a size class. The slots of the same machine form a clique: There
is an edge from each slot to the other with weight $0$. Furthermore,
there is an edge of weight $1$ from slot $u$ to $v$, when (1) $u$
and $v$ are not on the same machine, (2) $u$ and $v$ belong to
the same size class, and (3) $u$ is currently assigned a larger job
than $v$. Additionally, there is an edge of weight $0$ from the source
to every slot on a machine with load more than $u_i+\epsilon\cdot p_{\max}$. The algorithm performs a BFS on the graph above starting in $s$. Once it reaches a machine with load at most $u_i$, it selects the edge $(u,v)$ over which the machine was reached and swaps the jobs assigned to the slots $u$ and $v$. This is continued until every machine $i$ is assigned a load at most $u_i + \epsilon\cdot p_{\max}$.
We first show, that the distance from $s$ to any slot in the graph does not decrease by a swap.

\begin{lem}\label{lem:ls_swapdis}
The distance from $s$ to any slot does not decrease by a swap. 
\end{lem} 
Using Lemma~\ref{lem:ls_swapdis} and a potential function that is bounded by $n^3$ we show the following.

\begin{lem}\label{lem:ls_term}
The first stage of the algorithm terminates after at most $O(n^{3})$ swaps. 
\end{lem} 
The detailed proof of Lemmas~\ref{lem:ls_swapdis} and~\ref{lem:ls_term} are moved to Appendix~\ref{apx:det_ls}. Additionally, the counterparts for the second stage are shown (see Lemmas~\ref{lem:ls_swapdis_sec} and~\ref{lem:ls_term_sec}).  It is important to note that in the second stage the goal is to fix machines that have load less than $\ell_i - \epsilon\cdot p_{\max}$. We do so by swapping larger jobs on machines with load at least $\ell_i$ with smaller jobs (of the same type) on underloaded machines. As we never increase the load of a machine that has load at least $\ell_i - \epsilon\cdot p_{\max}$ this process does not lead to any violations of the machine upper bounds. Hence, Lemmas~\ref{lem:ls_swaps}-\ref{lem:ls_term} and Lemmas~\ref{lem:ls_swapdis_sec} and~\ref{lem:ls_term_sec} give the following result on the running time and additive approximation guarantee of the local search algorithm.

\begin{lem}
\label{lem:rounding-relaxation}Given a solution to the slot-MILP, in time
$n^{O(1)}$ we can compute an integral solution to slot-MILP such that
$\sum_{j\in\jobs}p_{j}x_{i,j}\in[\ell_{i}-\epsilon\cdot p_{\max},u_{i}+\epsilon\cdot p_{\max}]$
for each machine $i\in\machines$.
\end{lem}
By combining Lemmas~\ref{lem:solve-ILP-faster} and \ref{lem:rounding-relaxation}
we obtain our main theorem. Our running time is polynomial if $K=O(1)$
which will be the relevant case for our applications for makespan
minimization, Santa Claus on identical machines, etc.
\begin{thm}
\label{thm:main-thm}There is an algorithm for the target load balancing
problem with a running time of $m^{K+1}n^{O(1/\epsilon)}$
that computes a solution in which the load of each machine $i\in\machines$
is in $[\ell_{i}-\epsilon\cdot p_{\max},u_{i}+\epsilon\cdot p_{\max}]$,
or asserts that there is no feasible solution.
\end{thm}

\section{Applications}

We can use Theorem~\ref{thm:main-thm} in order to obtain additive
approximation schemes for \emph{makespan minimization}, the \emph{Santa Claus problem} and the \emph{envy-minimizing Santa Claus} problem on identical machines. The idea is to guess the target load intervals up to multiples of $\epsilon \cdot p_{\max}$ and then applying the algorithm due to Theorem~\ref{thm:main-thm}. As in these problems each machine has the same target load interval we have $K=1$. 

For $P||C_{\max}$ and the \emph{Santa Claus} problem the guessing procedure can be done in $O(1/\epsilon)$. As every machine is assigned the same target load interval,i.e. $\ell_i = \ell$ and $u_i = u$, we only need to guess the upper bound for $P||C_{\max}$ and lower bound for the \emph{Santa Claus} problem. For $P||C_{\max}$ we set $\ell = 0$ and guess $u$ as a multiple of $\epsilon \cdot p_{\max}$ within the interval $[\frac{1}{m}\sum_{j=1}^np_j,\frac{1}{m}\sum_{j=1}^np_j+p_{\max}]$. For the \emph{Santa Claus} problem we set $u = \sum_{j=1}^n p_j$ and guess $\ell$ within the interval $[\frac{1}{m}\sum_{j=1}^np_j-p_{\max},\frac{1}{m}\sum_{j=1}^np_j]$.

\begin{cor}
There is an algorithm for $P||C_{\max}$ with a running time of $m^2n^{O\left(1/\epsilon\right)}$
that computes a solution with makespan at most $OPT+\epsilon\cdot p_{\max}$.
\end{cor}

\begin{cor}
There is an algorithm for the Santa Claus problem on identical machines
with a running time of $m^2n^{O\left(1/\epsilon\right)}$ that computes
a solution in which each machine has a load of at least $OPT-\epsilon\cdot p_{\max}$.
\end{cor}
For the \emph{envy-minimizing} Santa Claus problem we need to guess both $\ell$ and $u$ simultaneously from the same intervals above. This gives a total number of guesses of $O(1/\epsilon^2)$.
\begin{cor}
There is an algorithm for envy-minimization on identical machines
with a running time of $m^2n^{O\left(1/\epsilon\right)}$ which computes
a solution with envy at most $OPT+\epsilon\cdot p_{\max}$.
\end{cor}

\section{Conclusion}

In this paper we introduced the concept of additive approximation schemes, where the goal is
to design algorithms that find solutions where the absolute deviation of the objective value from the optimal value is at most $\epsilon\cdot h$
for a (natural) parameter $h$
of the respective problem. Additive approximation schemes are interesting in among others the following situations. First, when the underlying problem does not admit a (traditional) PTAS, or even any (multiplicative) approximation guarantee, e.g., in the case that the optimal value is $0$, additive approximation schemes can give an alternative notion for approximating the problem. Secondly, in case that $h$ provides a lower bound on the optimum, then an additive approximation scheme immediately implies a (traditional) approximation scheme. However, when $h \ll \OPT$, then the additive approximation scheme can be much stronger.
We applied this concept to load balancing problems on identical machines finding an additive PTAS for makespan minimization, the Santa Claus problem and the envy-minimizing Santa Claus problem. We do so by introducing a new relaxation, the slot-MILP, and showing how to solve and round this relaxation. The running time of our method is exponential in both $1/\epsilon$ and the number of different target load intervals $K$. 

Therefore, we leave two open questions with respect to additive approximation schemes for the load balancing problems considered here. Firstly, for $P||C_{\max}$ there is an EPTAS~\cite{AlonEtal2009,JansenEtal2020}, i.e., an algorithm with
a running time of the form $f(1/\epsilon)\cdot n^{O(1)}$ for some
function $f$. We leave as an open question to find an additive approximation
scheme for the problem with such a running time or rule out that one
exists. Secondly, we leave open to find an additive approximation scheme for the
target load balancing problem when the number $K$ of different machine
types is super-constant. Note that our rounding algorithm from Section~\ref{sec:rounding-relaxation}
works for arbitrary $K$, but it is not clear whether one can solve the
slot-MILP in this case (approximatively) in polynomial time.  When the number of machines is not part of the input, then $Pm||C_{\max}$ admits an FPTAS, i.e., an approximation scheme with running time polynomial in the input size and $1/\epsilon$~\cite{Sahni1976}. Using the ideas in this paper, it can be easily shown that all three versions with a constant number of machines admit an additive FPTAS. On the other hand, when the number of machines is part of the input, then using similar arguments as in~\cite{GareyJohnson1978}, we can show that there cannot exist an additive FPTAS unless $\text{P} = \text{NP}$, as these problems are strongly NP-hard. For the case of unrelated machines, that is when the processing times, now denoted by $p_{ij}$, depend on the machine as well as the job, it can be shown that, unless $\text{P} = \text{NP}$, for each of the objectives considered in this paper, there does not exist an additive approximation algorithm with guarantee less then $\frac{1}{3} p_{\max}$, using the reduction from~\cite{LenstraEtal1990}. Therefore, the existence of an additive PTAS for this problem is also ruled out under the assumption that $\text{P} \neq \text{NP}$.

Another interesting direction of future research is to study the concept of additive approximation schemes for other types of problems.

\section*{Acknowledgements}

We wish to thank Jos\'e Verschae, Alexandra Lassota and Klaus Jansen for helpful discussions
on this problem. 
\bibliographystyle{plain}
\bibliography{main}
\appendix
\section{Detailed proofs of Section~\ref{sec:relax}}
\label{apx:relax_details}
\subsection{Proof of Lemma~\ref{lem:structural-property}}\label{apx:pf_strucprop}
\begin{proof} Let $(x,y)$ be a solution to the slot-MILP and assume that $x$ is the solution which minimizes 
\begin{equation}
\sum_{i\in\machines}\sum_{k=1}^{1/\epsilon}\lVert y_{i,k}\rVert_{2}.\label{potential}
\end{equation}
Now suppose toward contradiction that there are $i_{1}$, $i_{2}$
with $y_{i_{1}}\equiv y_{i_{2}}\bmod2$, but $y_{i_{1}}\neq y_{i_{2}}$.
We construct a new solution $x'$, which has a lower value of (\ref{potential}).
We set $x'_{i,j}=x_{i,j}$ for all $i\notin\{i_{1},i_{2}\}$ and $x'_{i_{1},j}=x'_{i_{2},j}=(x_{i_{1},j}+x_{i_{2},j})/2$.
In other words, we evenly distribute all jobs between $i_{1}$ and
$i_{2}$. Let us first check that the solution remains feasible. Let
$j\in\jobs$. Then 
\begin{align*}
\sum_{i\in\machines}x'_{i,j} & =x'_{i_{1},j}+x'_{i_{2},j}+\sum_{i\notin\{i_{1},i_{2}\}}x'_{i,j}\\
 & =\frac{x_{i_{1},j}+x_{i_{2},j}}{2}+\frac{x_{i_{1},j}+x_{i_{2},j}}{2}+\sum_{i\notin\{i_{1},i_{2}\}}x_{i,j}\\
 & =\sum_{i\in\machines}x_{i,j}=1.
\end{align*}
For all machines $i\notin\{i_{1},i_{2}\}$ the load does not change and, hence, the load of machine $i$ remains within $[\ell_i,u_i]$. For $i_{1}$ and $i_{2}$, we argue
\begin{align*}
\sum_{j\in\jobs}p_{j}x'_{i_{1},j}=\sum_{j\in\jobs}p_{j}x'_{i_{2},j} & =\sum_{j\in\jobs}p_{j}\frac{x_{i_{1},j}+x_{i_{2},j}}{2}\\
 & =\frac{1}{2}\sum_{j\in\jobs}p_{j}x_{i_{1},j}+\frac{1}{2}\sum_{j\in\jobs}p_{j}x_{i_{2},j}\\
 & \le\frac{u_i}{2}+\frac{u_i}{2}=u_i
\end{align*}
and 
\begin{align*}
\sum_{j\in\jobs}p_{j}x'_{i_{1},j}=\sum_{j\in\jobs}p_{j}x'_{i_{2},j} & =\sum_{j\in\jobs}p_{j}\frac{x_{i_{1},j}+x_{i_{2},j}}{2}\\
 & =\frac{1}{2}\sum_{j\in\jobs}p_{j}x_{i_{1},j}+\frac{1}{2}\sum_{j\in\jobs}p_{j}x_{i_{2},j}\\
 & \ge\frac{\ell_i}{2}+\frac{\ell_i}{2}=\ell_i.
\end{align*}
Hence, the solution remains optimal. As for the integrality constraints,
again the machines $i\notin\{i_{1},i_{2}\}$ do not change. Let $k\in\{1,\dotsc,1/\epsilon\}$.
Since $y_{i_{1},k}\equiv y_{i_{2},k}$, we have that $y_{i_{1},k}+y_{i_{2},k}$
is even. It follows that 
\[
\sum_{j\in\jobs_{k}}x'_{i_{1},j}=\sum_{j\in\jobs_{k}}x'_{i_{2},j}=\frac{y_{i_{1}}+y_{i_{2}}}{2}
\]
is integral. Now it remains to show that (\ref{potential}) has decreased.
Notice that by triangle inequality 
\[
\lVert y'_{i_{1},k}\rVert_{2}+\lVert y'_{i_{2},k}\rVert_{2}=2\left\lVert \frac{y_{i_{1},k}+y_{i_{2},k}}{2}\right\rVert _{2}\le\lVert y_{i_{1},k}\rVert_{2}+\lVert y_{i_{2},k}\rVert_{2}
\]
and strict inequality holds when $y_{i_{1},k}\neq y_{i_{2},k}$. Since
this is the case for at least one $k$ and all machines $i\notin\{i_{1},i_{2}\}$
do not change, we have that (\ref{potential}) has decreased. A contradiction.
\end{proof}

\subsection{Details on Dynamic Program of Section~\ref{sec:dp}}
\label{apx:dp_details}
Here we present detailed proofs of Lemmas~\ref{lem:structural-property-2}-\ref{lem:solve-ILP-faster} and a detailed description of the dynamic program.

\begin{proof}[Proof of Lemma~\ref{lem:structural-property-2}]

Condition~\eqref{eq:min-jobs-bound} and~\eqref{eq:volume-assigned} follow directly from feasibility of the solution.

To show condition~\eqref{eq:monotonicity}, let $x,y$ be a solution with corresponding average load vector $\left\{ z_{i,k}\right\} _{i\in\machines,k\in1,\dotsc,1/\epsilon\}}$, where the values $z_{i,k}$ when $y_{i,k} = 0$ are chosen appropriately. Let $\hat{z}_1 \le \cdots \le \hat{z}_{\bar{n}}$ be an ordering of the $\bar{n} = | \{ (i,k) : y_{ik} > 0 \} |$
values $z_{i,k}$ for all 
$i\in\machines,k\in\{1,\dotsc,1/\epsilon\}$ with $y_{ik} > 0$. Assume that $(x,y)$ is the solution maximizing the following potential function
\begin{equation}
\sum_{i=1}^{\bar{n}} n^{2(m/\epsilon-i)}\hat{z}_i \label{eq:pot}
\end{equation}
We will now show that in this case we can iteratively find an ordering of machines such that condition~\eqref{eq:monotonicity} holds and otherwise get a contradiction with respect to the potential function. Let $i$ be the machine minimizing $\sum_{k=1}^{1/\epsilon} z_{i,k}$. All other machines $i'$ must satisfy one of the following two cases: (1) $z_{i,k} \leq z_{i',k}$ for all $k \in \{1,...,1/\epsilon\}$ or (2) $z_{i,k} \andy{>} z_{i',k}$ for some $k$. If (1) holds for all machines $i'$ we relabel machine $i$ as machine $1$. Otherwise, let $i' \neq i$ be a machine such that for some $k$ 
\begin{equation}
z_{i,k}>z_{i',k}\label{order-k}
\end{equation}
Then, as $i$ minimizes $\sum_{k=1}^{1/\epsilon} z_{i,k}$ we know that there must exist $\overline{k} \neq k$ with 
\begin{equation}
 z_{i,\overline{k}}<z_{i',\overline{k}}\label{order-overlinek}
\end{equation}
As we can freely choose the value of $z_{\bar{i},k'}$, whenever $y_{\bar{i},k'} = 0$, we know that $y_{i,k}, y_{i,\overline{k}}, y_{i',k}, y_{i',\overline{k}} > 0$.
We now gradually exchange jobs of $\jobs_{k}$ and $\jobs_{\overline{k}}$
between $i$ and $i'$ without changing the total load on either of
the machines. Indeed, there must be some $j,j'\in\jobs_{k}$ with
$x_{i,j}>0$, $x_{i',j'}>0$, and $p_{j}>p_{j'}$. Conversely, there
are $\overline{j},\overline{j}'\in\jobs_{\overline{k}}$ with $x_{i,\overline{j}}>0$,
$x_{i',\overline{j}'}>0$, and $p_{\overline{j}}<p_{\overline{j}'}$.
For some $\delta, \overline{\delta}>0$ we now augment the solution
in the following way. 
\begin{align*}
x_{i,j'} & \gets x_{i,j'}+\delta & x_{i,\overline{j}'} & \gets x_{i,\overline{j}'}+\overline{\delta}\\
x_{i,j} & \gets x_{i,j}-\delta & x_{i,\overline{j}} & \gets x_{i,\overline{j}}-\overline{\delta}\\
x_{i',j'} & \gets x_{i',j'}-\delta & x_{i',\overline{j}'} & \gets x_{i',\overline{j}'}-\overline{\delta}\\
x_{i',j} & \gets x_{i',j}+\delta & x_{i',\overline{j}} & \gets x_{i',\overline{j}}+\overline{\delta}
\end{align*}
It is easy to see that for $\delta$ and $\overline{\delta}$ sufficiently
small each variable remains non-negative. Moreover, each job remains
fully assigned and the number of jobs of $\jobs_{k}$ and $\jobs_{\overline{k}}$
assigned to $i$ and $i'$ remains the same. 
By setting $\overline{\delta}=\delta(p_{j}-p_{j'})/(p_{\overline{j}'}-p_{\overline{j}})$ the load over each of the two machines stays the same. Furthermore, as $p_j > p_{j'}$ and $p_{\overline{j}'}>p_{\overline{j}}$ we have that $\delta, \overline{\delta}>0$. We choose $\delta$ maximal such that all
$x$ variables remain non-negative and the inequalities (\ref{order-k})
and (\ref{order-overlinek}) still hold or turn to equality. This means that we decreased $z_{i,k}$ by $\frac{\delta(p_j - p_j')}{y_{i,k}}$ and $z_{i',\overline{k}}$ by $\frac{\delta(p_{j}-p_{j'})}{y_{i,\overline{k}}}$. At the same time we increased $z_{i,\overline{k}}$ by $\frac{\delta(p_{j}-p_{j'})}{y_{i,\overline{k}}}$ and $z_{i',k}$ by $\frac{\delta(p_j - p_j')}{y_{i',k}}$. 
Since $z_{i',k}$ and $z_{i,\bar{k}}$ (the respective smaller $z$-variables for $i$ and $i'$ that we change) increase by at least $\frac{\delta(p_j - p_j')}{n}$ and $z_{i,k}$ and $z_{i',\bar{k}}$ decrease by at most $\delta(p_j - p_j')$,
we have that~\eqref{eq:pot} increases. This gives a contradiction.

As we can repeat this argument iteratively assuming that machines $\{1,\dots,i_0\}$ are correctly sorted for some $i_0 \in \{1,\dots,m\}$, we have that there exists a solution $(x,y)$ with vector $\left\{ z_{i,k}\right\} _{i\in\machines,k\in1,\dotsc,1/\epsilon\}}$ such that condition~\eqref{eq:monotonicity} holds.

\end{proof}


\begin{proof}[Proof of Lemma~\ref{lem:sufficient-conditions}]
We first show that there exists an assignment vector $\left\{ x_{i,j}\right\} _{i\in\machines,j\in\jobs}$ satisfying
\begin{equation}
\sum_{j\in\jobs_{k}}p_{j}x_{i,j}\le y_{i,k}z_{i,k}\label{property-2-desired}
\end{equation}
for all $i \in \machines$ and $k \in \{1,\dots,1/\epsilon\}$. In order to do so we use condition \eqref{eq:min-jobs-bound} of Lemma \ref{lem:structural-property-2}. We find this assignment independently for all $k$.
We start by assigning $\jobs_{k}^{\min}(y_{1,k})$ (completely) to
machine $1$, then $\jobs_{k}^{\min}(y_{1,k}+y_{2,k})\setminus\jobs_{k}^{\min}(y_{1,k})$
to machine $2$, etc. This assignment does not necessary have the
desired property (\ref{property-2-desired}). Hence, we repair the
property iteratively for $i=2,\dotsc,m$. Machine $1$ clearly satisfies
~\eqref{property-2-desired} because of~\eqref{eq:min-jobs-bound}). Let $i\in\{2,\dotsc,m-1\}$
such that all machines $1,\dotsc,i$ satisfy~\eqref{property-2-desired}.
In each iteration $i$ we do not touch any of the machines $i+1,\dotsc,m$.
Hence, when repairing machine $i$ we may assume that machines $1,\dotsc,i$
contain only $\jobs_{k}^{\min}(y_{1,k}+\cdots+y_{i,k})$. If machine
$i$ satisfies (\ref{property-2-desired}) we are done and continue
with $i+1$. Otherwise, we know that there is a job $j$ with $p_{j}>z_{i,k}$
and $x_{i,j}>0$. Moreover, because of condition \eqref{eq:min-jobs-bound} we have 
\begin{equation}
\sum_{i'=1}^{i}\sum_{j\in\jobs_{k}^{\min}(y_{1,k}+\cdots+y_{i,k})}p_{j}x_{i',j}=\sum_{j\in\jobs_{k}^{\min}(y_{1,k}+\cdots+y_{i,k})}p_{j}\le y_{1,k}z_{1,k}+\cdots+y_{i,k}z_{i,k}.
\end{equation}
Since $i$ violates (\ref{property-2-desired}) there must be some
$i'<i$ satisfying (\ref{property-2-desired}) with strict inequality.
In particular, there is a job $j'$ with $p_{j'}<z_{i',k}\le z_{i,k}<p_{j}$
and $x_{i',j'}>0$. We now choose an $\alpha > 0$ and exchange $j'$ and $j$ between $i$ and $i'$ as follows
\begin{align*}
x_{i,j'} & \gets x_{i,j'}+ \alpha & x_{i',j'} & \gets x_{i',j'} - \alpha\\
x_{i,j} & \gets x_{i,j'} -\alpha & x_{i',j'} & \gets x_{i',j'} + \alpha\\
\end{align*}
Clearly, the solution remains feasible. We choose $\alpha$ maximal such that either $i'$ satisfies (\ref{property-2-desired}) with equality, $i$ satisfies (\ref{property-2-desired}), $x_{i,j}=0$, or $x_{i',j'}=0$. The choice of $\alpha$ makes sure that each pair $j,j'$ that can be exchanged like this will only be exchanged once.  This procedure is repeated until $i$ satisfies (\ref{property-2-desired}). As the procedure is repeated for all $i$ and possibly has to check all pairs of every job type in each exchange we have a running time of $O(mn^2)$.

Next, we claim that for all $i$ and $k$, we have
\begin{equation}
    \sum_{j \in \jobs_k}p_jx_{i,j} \ge y_{i,k}z_{i,k} - \delta \epsilon\cdot p_{\max} \label{eq:machine_under}
\end{equation}
To prove this claim, assume by contradiction that for some machine $i'$~\eqref{eq:machine_under} does not hold. Then by~\eqref{property-2-desired} and condition~\eqref{eq:volume-assigned}, we have that
\begin{equation}
    \sum_{j \in \jobs_k} \sum_{i \in \machines} x_{i,j}p_j < \sum_{i \in \machines} y_{i,k}z_{i,k} - \delta \epsilon\cdot p_{\max} < \sum_{j \in \jobs_k}p_j
\end{equation}
This contradicts the fact that all jobs are fully assigned and thus $\sum_{j \in \jobs_k} p_j x_{i,j} = \sum_{j \in \jobs_k} p_j$. Hence, we have that 
\begin{equation}
    \sum_{k=1}^{1/\epsilon} \sum_{k \in \jobs_k} p_j x_{ij} \geq \sum_{k=1}^{1/\epsilon} y_{i,k} z_{i,k} - \delta p_{\max} \geq \ell_i - \delta p_{\max},
\end{equation}
where the last inequality follows by condition~\eqref{eq:bounds}.
\end{proof}

Based on Lemmas~\ref{lem:structural-property-2} and~\ref{lem:sufficient-conditions}, the goal of our DP is to compute vectors $\left\{ y_{i,k},z_{i,k}\right\} _{i\in\machines,k\in\{1,\dotsc,1/\epsilon\}}$
that satisfy the conditions due to Lemma~\ref{lem:sufficient-conditions}.
The key insight is now that when we consider the next machine $i'$
in the ordering, we do not need to remember all vectors $\left\{ y_{i,k},z_{i,k}\right\} _{i\in\machines,k\in\{1,\dotsc,1/\epsilon\}}$
for all previously considered machines~$i$, but it suffices to remember
the number of previously assigned jobs from each set $\jobs_{k}$,
the current left hand side of inequality~\eqref{eq:min-jobs-bound}
for each $k$, the vector $\left\{ z_{i'',k}\right\} _{k\in\{1,\dotsc,1/\epsilon\}}$
of the previously considered machine $i''$, and for each type of
machines the number of previously guessed machines of this type. We
say that two machines $i,i'\in\machines$ are of the same \emph{type
}if $[\ell_{i},u_{i}]=[\ell_{i'},u_{i'}]$. Recall that $K$ denotes
the number of different types of machines. Let $\K:=\left\{ [\ell_{i},u_{i}]|i\in\machines\right\} $
and define $\ell^{(1)},...,\ell^{(|\K|)},u^{(1)},...,u^{(|\K|)}$
such that $\K=\left\{ \left[\ell^{(1)},u^{(1)}\right],...,\left[\ell^{(|\K|)},u^{(|\K|)}\right]\right\} $.
For each $r\in\{1,...,K\}$ let $m_{r}$ denote the number of machines
$i\in\machines$ such that $[\ell_{i},u_{i}]=\left[\ell^{(r)},u^{(r)}\right]$.
We say that such a machine is of \emph{type} $r$.

We introduce a DP-table with one cell for each combination of
\begin{itemize}
\item a value $i$ indicating the number of machines that have already been considered with $i \in \{0,...,m\}$,
\item a value $m_{r}'\in\{0,...,m_{r}\}$ for each interval $r\in\{1,...,|\K|\}$
indicating the number of machines of type $r$ for which we have already
defined vectors $\left\{ y_{i',k},z_{i',k}\right\} _{k\in\{1,\dotsc,1/\epsilon\}}$,
let $\machines'$ denote these machines intuitively,
\item a vector $\left\{ z_{i,k}\right\} _{k\in\{1,\dotsc,1/\epsilon\}}$ where
$z_{i,k}\in\left\{ 0,\frac{\delta\epsilon}{n}p_{\max},\frac{2\delta\epsilon}{n}p_{\max},...,p_{\max}\right\} $
for each $k\in\{1,\dotsc,1/\epsilon\}$. The vector~$z_{i,k}$ corresponds to the average loads on the currently considered machine,
\item a vector $\left\{ y_{i,k}\right\} _{k\in\{1,\dotsc,1/\epsilon\}}$ where
$y_{i,k}\in\left\{ 0,1,2,...,n_k\right\} $
for each $k\in\{1,\dotsc,1/\epsilon\}$. The vector~$y_{i,k}$ corresponds to the number of jobs on the currently considered machine,
\item for each $k\in\{1,...,1/\epsilon\}$
\begin{itemize}
\item a value $n'_{k}\in\{0,...,|\jobs_{k}|\}$ indicating the number of
jobs of in $\jobs_{k}$ that were previously assigned to machines
in $\machines'$,
\item a value $S_{k}$ with $S_{k} \in\left\{ 0,\frac{\delta\epsilon}{n}p_{\max},\frac{2\delta\epsilon}{n}p_{\max},...,np_{\max}\right\}$ which corresponds to the value $\sum_{i'=1}^{i}y_{i',k}z_{i',k}$.
\end{itemize}
\end{itemize}
Each cell corresponds to the subproblem of checking whether there is a solution using machines $\{1,\dots,i\}$ such that each machine type $r \in \{1,\dots, K\}$ is used $m_r$ times, machine $i$ is assigned $y_{i,k}$ jobs of each type $k$ with average load $z_{i,k}$, for each type $k$ a total of $n'_k$ jobs is assigned and the total volume assigned of each type $k$ is $S_k$. Due to the dimension of the values corresponding to a DP-cell, the dimension of the DP table is given by $m^{K+1}\cdot (\frac{n}{\delta\epsilon})^{O(1/\epsilon)}$. 

When considering cell 
\[
\left(i,\left\{ m_{r}\right\} _{r\in\{1,...,|\K|\}},\left\{ z_{i,k}\right\} _{k\in\{1,\dotsc,1/\epsilon\}},\left\{ y_{i,k}\right\} _{k\in\{1,\dotsc,1/\epsilon\}},\left\{ n'_{k}\right\} _{k\in\{1,\dotsc,1/\epsilon\}},\left\{ S_{k}\right\} _{k\in\{1,\dotsc,1/\epsilon\}}\right)
\]
the DP proceeds as follows: for every machine type $r^* \in \{1,...,K\}$ it checks whether for some $\left\{ \tilde{z}_{i-1,k}\right\} _{k\in\{1,\dotsc,1/\epsilon\}}$ and $\left\{ \tilde{y}_{i-1,k}\right\} _{k\in\{1,\dotsc,1/\epsilon\}}$, the value of the following cell is true
\[
\left(i-1,\left\{ \tilde{m}_{r}\right\} _{r\in\{1,...,|\K|\}},\left\{ \tilde{z}_{i-1,k}\right\} _{k\in\{1,\dotsc,1/\epsilon\}},\left\{ \tilde{y}_{i-1,k}\right\} _{k\in\{1,\dotsc,1/\epsilon\}},\left\{ \tilde{n}'_{k}\right\} _{k\in\{1,\dotsc,1/\epsilon\}},\left\{ \tilde{S}_{k}\right\} _{k\in\{1,\dotsc,1/\epsilon\}}\right),
\]
where $\tilde{m}_{r}=m_{r}$ for each $r\ne r^{*}$ and $\tilde{m}_{r^{*}}=m_{r^{*}}-1$,
$\tilde{n}'_{k}=n'_{k}-y_{i,k}$ for each $k\in\{1,\dotsc,1/\epsilon\}$,
and $\tilde{S}_{k}=S_{k}-y_{i,k}z_{i,k}$ for each $k\in\{1,\dotsc,1/\epsilon\}$. Then we need to check if the following conditions are true for all $k \in \{1,...,1/\epsilon\}$
\begin{equation}
    S_{k}\le\delta\epsilon\cdot p_{\max}+\sum_{j\in\jobs_{k}}p_{j} \label{eq:dp_con1}
\end{equation}
\begin{equation}
    z_{i,k} \geq \tilde{z}_{i-1,k} \label{eq:dp_con2}
\end{equation}
If these conditions are true, then there exists a solution corresponding to the considered DP cell. Filling each cell takes $K \cdot (\frac{n}{\delta\epsilon})^{O(1/\epsilon)}$. 

Finally, for each possible value of $\left\{ z_{m,k}\right\} _{k\in\{1,\dotsc,1/\epsilon\}}$ and $\left\{ m_{m,k}\right\} _{k\in\{1,\dotsc,1/\epsilon\}}$ we check whether there exists a solution for the DP cell
\[
\left(m,\left\{ m_{r}\right\} _{r\in\{1,...,|\K|\}},\left\{ z_{m,k}\right\} _{k\in\{1,\dotsc,1/\epsilon\}},\left\{ y_{m,k}\right\} _{k\in\{1,\dotsc,1/\epsilon\}},\left\{ n'_{k}\right\} _{k\in\{1,\dotsc,1/\epsilon\}},\left\{ S_{k}\right\} _{k\in\{1,\dotsc,1/\epsilon\}}\right)
\]
where the correct number of each machine type is considered, all jobs are assigned and for every $k \in \{1,\dots,1/\epsilon\}$ we have that 
\[
\sum_{j \in \jobs_k} p_j \le S_k \le \sum_{j \in \jobs_k} p_j + \delta\epsilon\cdot p_{\max}.
\]
If this is the case we use standard backward recursion to find vectors $\left\{ z_{i,k}\right\} _{k\in\{1,\dotsc,1/\epsilon\}}$ and $\left\{ y_{i,k}\right\} _{k\in\{1,\dotsc,1/\epsilon\}}$ for all $i \in \{1,...,m\}$ and an assignment of machine types to machine indices and use Lemma~\ref{lem:sufficient-conditions} to obtain a solution $(x,y)$ to slot-MILP'. If there is no such solution we assert that there is no solution to the original relaxation, i.e., to slot-MILP.

\begin{proof}[Proof of Lemma~\ref{lem:solve-ILP-faster}]
The running time follows from the dimension of the DP table and the time it takes to validate a specific DP cell. This amounts to a running time of $m^{K+1}(\frac{n}{\delta\epsilon})^{O(1/\epsilon)}$. 

For the correctness of the DP we need two observations: (1) due to conditions~\eqref{eq:dp_con1} and~\eqref{eq:dp_con2} and the way we check whether a solution corresponding to a DP cell exists we have that there exists a solution for machine $i$ if and only if there is a solution for machine $i-1$. Hence, we can indeed find a solution via backward recursion and (2) if for some $\left\{ z_{m,k}\right\} _{k\in\{1,\dotsc,1/\epsilon\}}$ and $\left\{ y_{m,k}\right\} _{k\in\{1,\dotsc,1/\epsilon\}}$ there is a solution for the cell
\[
\left(m,\left\{ m_{r}\right\} _{r\in\{1,...,|\K|\}},\left\{ z_{m,k}\right\} _{k\in\{1,\dotsc,1/\epsilon\}},\left\{ y_{m,k}\right\} _{k\in\{1,\dotsc,1/\epsilon\}},\left\{ n'_{k}\right\} _{k\in\{1,\dotsc,1/\epsilon\}},\left\{ S_{k}\right\} _{k\in\{1,\dotsc,1/\epsilon\}}\right),
\]
then we can apply Lemma~\ref{lem:sufficient-conditions} to find a solution to slot-MILP'. If there is no such solution, we know that due to our rounding of the $z$-values there is also no solution satisfying Lemma~\ref{lem:structural-property-2}. This implies that there is no solution to the slot-MILP.
\end{proof}

\section{Details on the local search algorithm}
\label{apx:det_ls}
 We first show that Lemma~\ref{lem:ls_swaps} also holds for the second stage of the algorithm.
\begin{proof}[Proof of Lemma~\ref{lem:ls_swaps} for the second stage]
 We want to find a job $j$ of type $k$ assigned to a machine with load less than $\ell_i$ and swap it with some job $j'$ of the same type with $p_j' > p_j$ currently assigned to machine $i'$ with load at least $\ell_i'$. Suppose towards contradiction that no such pair of jobs exists. Then we know that the machines in $\machines_1 \cup \dots \cup \machines_\ell$ are assigned the largest jobs of type $k$ while having a load less than their specific lower bounds. This implies that even in a fractional solution the total load cannot be increased and, hence, at least one machine violates its lower target. This gives a contradiction.
 
 As every job is fully assigned to a machine we have that each job of a type $k$ which is assigned to a machine in the  sets $\machines_1,\machines_2,\dots,\machines_\ell$ is considered exactly once and compared to every other job of type $k$ not assigned to these machines. As the existence of a pair was shown for an arbitrary $k$ this gives a worst case running time of $O(n^2)$.
\end{proof}
Next, we present the proofs of Lemmas~\ref{lem:ls_swapdis} and~\ref{lem:ls_term}.
\begin{proof}[Proof of Lemma~\ref{lem:ls_swapdis}] 
We observe how the graph changes when swapping two jobs. Throughout the proof the distance $d(w)$
of a slot $w$ is the distance from $s$ to $w$ before the swap is
executed. Clearly, removing any edges from the graph cannot decrease
the distances of vertices. Edges between slots of the same machine
do not change and no new edges can be added from the source, since
during the execution of the algorithm a machine with load at most
$u_i+\epsilon\cdot p_{\max}$ will never exceed $u_i+\epsilon\cdot p_{\max}$. Hence,
it suffices to look at the changes in edges of weight $1$. Although
such an edge $(u,v)$ might be added to the graph, we will show that
this happens only when $d(v)\le d(u)+1$. Adding this edges cannot
decrease any distances, since the first part of a shortest path using
$(u,v)$ could always be replaced by a path to $v$ without this edge.

Now we have to check that these are the only changes made to the graph.
Let $u,v$ be the slots in which we exchange the jobs. The size of
the job in $u$ decreases; the size of the job in $v$ increases.
We only need to look at the incoming and outgoing edges of $u$ and
$v$, since all other edges remain the same.

Consider the incoming edges of $u$. Since the size of the job in
$u$ decreases, there could be new incoming edges. Let $(w,u)$ be
an edge of weight $1$ that is added. This means the job in slot $w$
has a larger size than the job on $u$. Either $w$ is a slot on the
same machine as $v$ or $(w,v)$ is in the graph before the swap.
The former case implies that $d(w)=d(v)=d(u)+1$. In the latter case
we have $d(u)=d(v)-1\le d(w)$. No outgoing edge from $u$ can be
added, since the size of $u$'s job decreases.

Now consider $v$. Since the size of its job increases, no incoming
edge can be added. As for the outgoing edges, let $(v,w)$ be an outgoing
edge added by the swap. Then either $u$ and $w$ are slots on the
same machine or $(u,w)$ was is in the graph before the swap. In the
former case, $d(w)=d(u)=d(v)-1$. In the latter case, $d(w)\le d(u)+1=d(v)$. 
\end{proof} 

\begin{proof}[Proof of Lemma~\ref{lem:ls_term}]
Let $j_{1},\dotsc,j_{n}$ be the jobs $\jobs$ in increasing order of size. We claim that the potential
\[
\sum_{i=1}^{n}i\cdot d(j_{i})
\]
increases with every swap. Here $d(j_{i})$ denotes the distance from
$s$ to the slot to which $j_{i}$ is assigned. Since the function
is integral and bounded by $n^{3}$, the claim follows. Let $j_{k}$,
$j_{h}$ be the jobs that are swapped. Assume that $k<h$, i.e.,
$p_{j_{k}}<p_{j_{h}}$. Let $d$, $d'$ be the distance functions
before and after the swap. Based on Lemma~\ref{lem:ls_swapdis} we have $d'(j_{i})\ge d(j_{i})$ for all $i\notin\{k,h\}$, $d'(j_{h})\ge d(j_{k})$
and $d'(j_{k})\ge d(j_{h})$, since these jobs swapped their slots.
It follows that 
\begin{equation*}
    \begin{split}
    k\cdot d'(j_{k})+h\cdot d'(j_{h}) & \ge k\cdot d(j_{h})+h\cdot d(j_{k})\\
    & =k\cdot d(j_{h})+(h-k)\cdot\underbrace{d(j_{k})}_{>d(j_{h})}+k\cdot d(j_{k})>h\cdot d(j_{h})+k\cdot d(j_{k}).
    \end{split}
\end{equation*}
\end{proof}
For the second stage, the analogues of Lemmas~\ref{lem:ls_swapdis} and~\ref{lem:ls_term} follow from similar arguments. We again construct a graph on $|\jobs| + 1$ vertices. The difference in the construction is that there is an edge from $s$ to every slot on a machine with load less than $\ell_i-\epsilon\cdot p_{\max}$ and there is an edge of weight $1$ from slot $u$ to $v$, when (1) $u$ and $v$ are not on the same machine, (2) $u$ and $v$ belong to the same size class, and (3) $u$ is currently assigned a smaller job than $v$. Again, the breadth-first search starts at $s$ and once a machine with load at least $\ell_i$ is reached the algorithm selects the jobs assigned to the vertices $u$ and $v$ corresponding to the edge $(u,v)$ over which the machine was reached and swaps them. This procedure is continued until for every machine $i$ the load is at least $\ell_i - \epsilon\cdot p_{\max}$. And due to Lemma~\ref{lem:ls_swaps} each iteration finishes in $O(n^2)$.

\begin{lem}\label{lem:ls_swapdis_sec}
The distance from $s$ to any slot does not decrease by a swap in the second stage of the algorithm. 
\end{lem} 
\begin{proof}
Let $d(w)$ denote the distance from $s$ to slot $w$ before a swap. Clearly, removing edges does not decrease the distance from $s$ to any vertex. Edges between slots of the same machine do not change and we do not add new edges from the source to a machine since we only decrease the load on machines with load at least $\ell_i$ and, hence, do not decrease any machine load below $\ell_i-\epsilon\cdot p_{\max}$. So we only have to consider the changes in edges of weight $1$. Such an edge $(u,v)$ might be added to the graph but only when $d(v) \leq d(u) + 1$. Adding this edge will not decrease any distances, since the first part of a shortest path using $(u,v)$ can always be replaced by a path to $v$ not using $(u,v)$. We have to check that these are the only changes made to the graph.

Let $u$ and $v$ be the slots in which we exchange the jobs. This implies that the size of the job assigned to slot $u$ increases and the size of the job assigned to slot $v$ decreases. Next we look at the incoming edges of $u$ and outgoing edges of $v$ (all others remain the same). As the size of the job assigned to slot $u$ increases there could be a new incoming edge from slots with smaller jobs. Let $(w,u)$ be such an edge. Then either slot $w$ is on the same machine as $v$ or $(w,v)$ was an edge with weight $1$ before the swap. The former case implies that $d(w) = d(v) = d(u) + 1$ and in the latter case $d(u) = d(v) - 1 \leq d(w)$. In both cases we have that $d(u) \leq d(w) + 1$ which satisfies the property above. Next consider the new outgoing edges of $v$. Let $(v,w)$ be such an edge. Then the size of the job assigned to $w$ is larger than the job assigned to $v$ which was originally assigned to $u$. So, either $u$ and $w$ are on the same machine and $d(w) = d(u) = d(v)-1$ or $(u,w)$ was an edge of weight $1$ in the original graph and $d(w) \leq d(u) + 1 = d(v)$.
\end{proof} 

\begin{lem}\label{lem:ls_term_sec}
The second stage of the algorithm terminates after at most $O(n^{3})$ swaps. 
\end{lem} 

\begin{proof}
 Let $j_1,\dots,j_n$ be the jobs in $\jobs$ in decreasing order of size. We claim that the potential
$$\sum_{i=1}^n i\cdot d(j_i)$$
increases with every swap. Let $j_k$ and $j_h$ be the jobs that were swapped and assume that $k<h$, i.e. $p_{j_k} > p_{j_h}$. This implies that the edge from the slot of $j_h$ to $j_k$ was deleted and new edges were constructed as described above. Let $d$ and $d'$ be the distances before and after the swaps, respectively. Based on Lemma \ref{lem:ls_swapdis} we have that $d'(j_i)\geq d(j_i)$ for all $i \notin \{h,k\}$, $d'(j_h) \geq d(j_k)$ and $d'(j_k) \geq d(j_h)$. Furthermore, we know $d(j_k) > d(j_h)$. From this it follows that 
\begin{equation*}
    \begin{split}
        kd'(j_k) + h d'(j_h) & \geq kd(j_h) + h d(j_k) \\
        & = kd(j_h) + (h-k)\cdot \underbrace{d(j_k)}_{> d(j_h)}  +  kd(j_k) \geq h d(j_h) + kd(j_k) 
    \end{split}
\end{equation*}
This concludes the proof.
\end{proof}
Hence, similar to the first stage, the second stage of the algorithm finishes in $n^{O(1)}$. 
\end{document}